\documentclass{article}
\usepackage{amsmath,amsthm,amssymb}
\usepackage{url,fullpage}
\newtheorem{theorem}{Theorem}
\newtheorem{proposition}{Proposition}
\newtheorem{lemma}{Lemma}
\theoremstyle{definition}
\newtheorem{definition}{Definition}
\theoremstyle{remark}
\newtheorem{remark}{Remark}
\newtheorem{example}{Example}
\begin{document}

\title{A Mathematical Problem For Security Analysis Of Hash Functions And Pseudorandom Generators}
\author{Koji Nuida\footnote{%
Combutational Biology Research Center (CBRC) / Research Institute for Secure Systems (RISEC), %
National Institute of Advanced Industrial Science and Technology (AIST), Japan %
(\texttt{k.nuida@aist.go.jp})%
}, Takuro Abe\footnote{%
Department of Mechanical Engineering and Science, Kyoto University, Japan %
(\texttt{abe.takuro.4c@kyoto-u.ac.jp})%
}, Shizuo Kaji\footnote{%
Department of Mathematical Sciences, Faculty of Science, Yamaguchi University, Japan %
(\texttt{skaji@yamaguchi-u.ac.jp})%
}, Toshiaki Maeno\footnote{%
Department of Mathematics, Meijo University, Japan %
(\texttt{tmaeno@meijo-u.ac.jp})%
}, Yasuhide Numata\footnote{%
Department of Mathematical Sciences, Shinshu University, Japan %
(\texttt{nu@math.shinshu-u.ac.jp})%
}}
\date{\today}

\maketitle

\begin{abstract}
In this paper, we specify a class of mathematical problems, which we refer to as \lq\lq Function Density Problems'' (FDPs, in short), and point out novel connections of FDPs to the following two cryptographic topics; theoretical security evaluations of keyless hash functions (such as SHA-1), and constructions of provably secure pseudorandom generators (PRGs) with some enhanced security property introduced by Dubrov and Ishai (STOC 2006).
Our argument aims at proposing new theoretical frameworks for these topics (especially for the former) based on FDPs, rather than providing some concrete and practical results on the topics.
We also give some examples of mathematical discussions on FDPs, which would be of independent interest from mathematical viewpoints.
Finally, we discuss possible directions of future research on other cryptographic applications of FDPs and on mathematical studies on FDPs themselves.
\end{abstract}

\section{Introduction}
\label{sec:intro}

\subsection{Background and related works}
\label{subsec:intro_background}

It is widely understood that some mathematical problems have been playing indispensable roles in research on cryptography and information security.
For instance, the (expected) computational difficulty of integer factorization is the source of security of RSA cryptosystem \cite{RSA78}, while the problem of solving multivariate quadratic (MQ) equations has attracted several studies after the development of Matsumoto-Imai cryptosystem \cite{MI88} and its variants, whose constructions are closely related to MQ equations.
Hence, posing and studying an interesting mathematical problem which arises in certain cryptographic settings can contribute to the progress of cryptography and information security.

The aim of this paper is to emphasize the significance of a certain mathematical problem, which has connections to the following two major topics in information security; security analysis of keyless hash functions in the real world (such as MD5 and SHA-1), and construction of pseudorandom generators (PRGs) with some enhanced security property.
First, we give some descriptions of these two topics.

\paragraph{Security analysis of keyless hash functions.}

Intuitively, a hash function is a function $H \colon X \to Y$ from some (finite) set $X$ to another (finite) set $Y$ that possesses a certain desirable security property.
When we concern efficiency or computability of $H$, we consider an algorithm that computes $H$ (also denoted by $H$) and call it a hash algorithm.
One of the standard security requirements for hash functions is \emph{collision resistance}, which informally means that it is difficult to find a collision pair $(x_1,x_2)$ for $H$, i.e., $x_1 \neq x_2 \in X$ satisfying $H(x_1) = H(x_2)$.
Hash functions have been playing central roles in various information security applications, and secure hash functions for real-life applications are usually expected to possess the collision resistance property.

However, most of the preceding successful studies that show security of hash functions actually dealt with \emph{keyed} hash functions (or hash \emph{families}); intuitively, a family of hash functions $H_k$ parameterized by a key $k$ is called collision resistant if, for any (efficient) adversary, the attack to find a collision pair of $H_k$ fails with high probability for a randomly chosen key $k$.
Several constructions of keyed hash functions have been proposed so far (e.g., \cite{CW77}).
The above security notion of keyed hash functions can be interpreted as allowing one to (randomly) choose a concrete instance $H_k$ of the hash family \emph{after} an adversary is given.
In contrast, in most of real-life applications, the concrete instance of hash algorithms is specified first (for example, by a standardization), and then an adversary can try to attack the fixed hash algorithm.
This reversal of order causes a crucial difficulty in guaranteeing (or even formalizing in a reasonable manner) security of a keyless hash algorithm $H$, as (unless the trivial situation where the domain of $H$ is not larger than the image of $H$) there \emph{does} always \emph{exist} a collision pair $(x_1,x_2)$ for $H$ and any adversary (existing in theory) who innately knows the pair $(x_1,x_2)$ is obviously able to efficiently attack the fixed hash algorithm $H$.
In fact, even an instance of standardized (or de facto standard) hash algorithms, whose security must be evaluated well before the standardization, has been suffered from feasible attacks (e.g., \cite{WYY05}).
In this paper, we try to propose a theoretical and unified way to say something, preferably affirmative, about security of a concrete (keyless) instance of hash algorithms.

For related works, Rogaway \cite{Rog06} gave a detailed observation about the difference between \lq\lq inexistence of effective attack algorithms'' and \lq\lq lack of knowledge on construction of effective attack algorithms'' for keyless hash algorithms.
He emphasized the difference of the two situations (by the term \lq\lq human ignorance''), and discussed how to prove security of a cryptographic protocol by reducing the security into \lq\lq lack of knowledge on concrete attacks'' on the hash algorithm internally used by the protocol.
However, he did not discuss how to theoretically evaluate security of keyless hash algorithms themselves, which we study in this paper.
On the other hand, in this paper we adopt concrete security formulation rather than asymptotic one; while some observation for security of keyless hash algorithms in asymptotic security formulation is also given in Rogaway's paper.

\paragraph{Construction of enhanced PRGs.}

A PRG is an algorithm $G \colon S \to X$ with (finite) set $S$ of inputs (\emph{seeds}) and (finite) output set $X$ with the property that, when a seed $s \in S$ is chosen uniformly at random, the output $G(s) \in X$ of $G$ is also \lq\lq random'' in some sense.
Conventionally, the meaning of \lq\lq randomness'' here is formulated by using the notion of \emph{distinguisher}, which is an algorithm $D \colon X \to \{0,1\}$ with $1$-bit output and the input set being the output set $X$ of $G$.
In this paper we adopt concrete security formulation rather than asymptotic one, in which case the security requirement for PRGs can be formulated as $(T,\varepsilon)$-security; namely, $G$ is called \emph{$(T,\varepsilon)$-secure} \cite{FSS07} if, for any distinguisher $D$ for $G$ with (time) complexity bounded by $T$, the statistical distance between the output distribution $D(G(U_S))$ of $D$ with input given by $G$ with uniformly random seed $s \in S$ (referred to as \lq\lq pseudorandom input'') and the output distribution $D(U_X)$ of $D$ with uniformly random input $x \in X$ (referred to as \lq\lq random input'') is bounded by $\varepsilon$.
(Intuitively, any such $D$ cannot distinguish the random element $x$ and the pseudorandom element $G(s)$ in $X$ with significant advantage.)
There are a large number of constructions of PRGs, most of which are provably secure (possibly in asymptotic security formulation) under standard computational assumptions (e.g., \cite{BBS86,FSS07}).

On the other hand, in a preceding work of Dubrov and Ishai \cite{DI06}, an enhanced notion for PRGs, called \emph{pseudorandom generators that fool non-boolean distinguishers} (\emph{nb-PRGs}, in short), was proposed.
This notion is obtained by allowing the distinguishers $D$ in the above security notion to have larger output sets; namely, $G$ is called \emph{$(T,n,\varepsilon)$-secure} if, for any \lq\lq non-boolean'' distinguisher $D \colon X \to Y$ for $G$ with (time) complexity bounded by $T$ \emph{and output set $Y$ of size at most $n$}, the statistical distance between the output distributions of $D$ with random and pseudorandom inputs is bounded by $\varepsilon$.
Dubrov and Ishai showed interesting applications of nb-PRGs, e.g., secure pseudorandomization of a certain kind of information-theoretically secure protocols without any restriction on computational complexity of the adversary's attack algorithm.

However, constructing secure nb-PRGs seems much more difficult than the case of the usual PRGs.
Indeed, to the authors' best knowledge, the only constructions of nb-PRGs proposed so far are ones in the original paper \cite{DI06}, which are based on certain less standard computational assumptions.
Hence it will be fruitful if we can give some results implying that \emph{any} usual PRG (with some parameter) is also an nb-PRG (with a possibly different parameter).
In fact, a straightforward implication has been mentioned in \cite{DI06}, but this is far from being efficient (i.e., to obtain nb-PRGs with reasonable security parameters, the original PRGs are required to have somewhat impractical security parameters).
In this paper, we try to establish a more efficient implication result.

\subsection{Our contributions, and organization of this paper}
\label{subsec:intro_contributions}

In Section \ref{sec:problem}, we propose a class of mathematical problems, which we refer to as \lq\lq Function Density Problems''.
Intuitively, this problem is to evaluate the possibility of close approximations of arbitrary functions by using some \lq\lq easily describable (or analyzable)'' functions.

Then we introduce motivating applications of Function Density Problems to two topics in information security.
First, in Section \ref{sec:hash_functions}, we discuss theoretical analysis of collision resistance of keyless hash algorithms.
We give an abstract framework for attacking a given hash algorithm by using known attacks on some other \lq\lq easily breakable'' hash algorithms.
In the framework, it is essential to evaluate how closely a target hash algorithm can be approximated by \lq\lq easily breakable'' hash algorithms; thus Function Density Problems play a significant role in the security evaluation of hash algorithms.

Secondly, in Section \ref{sec:nb-PRG}, we study an enhanced security notion for PRGs (called nb-PRG) introduced by Dubrov and Ishai \cite{DI06}.
We give some implication results showing that any secure PRG with some parameter is also a secure nb-PRG with somewhat modified security parameter.
In the results, the overheads in the bounds of (time) complexity and of advantages for the distinguishers are in trade-off relations, and Function Density Problems can be applied to evaluate to what extent the trade-off will be improved by our proposed result.

Then, in order to arise some image or intuition of how Function Density Problems can be mathematically studied, in Section \ref{sec:example_FDP} we give some concrete examples of mathematical discussions on Function Density Problems themselves, using combinatorial and geometric arguments and techniques in Gr\"{o}bner bases.
In particular, we deal with special cases where the set of \lq\lq easily describable (or analyzable)'' functions forms a linear subspace (related to low-degree boolean functions, perfect linear codes and Reed--Solomon codes), which would be of independent interest from mathematical viewpoints.

Finally, in Section \ref{sec:conclusion} we give a concluding remark, which includes discussions on further possible applications of Function Density Problems in information security, and on possible directions of future research on Function Density Problems themselves.

\section{Function Density Problems}
\label{sec:problem}

In this section, we specify a class of mathematical problems, which we call \emph{Function Density Problems} (FDPs) in this paper.
As the class of FDPs in a most general form will include too various problems to obtain meaningful insights for their properties, it is significant to restrict the class suitably according to each situation under consideration.
Relations of FDPs to some concrete topics in cryptography will be shown in the following sections.

We give a general description of our problem:
\begin{definition}
[Function Density Problems]
\label{defn:function_density_problem}
Let $\mathcal{C}$ be a set of some functions, and let $\mathcal{C}'$ be a subset of $\mathcal{C}$.
Let $d(\cdot,\cdot)$ be a distance function for the pairs of functions in $\mathcal{C}$.
In this setting, we define a \emph{Function Density Problem} to be a problem of estimating the following quantity:
\begin{equation}
\label{eq:function_density_problem__definition}
r(\mathcal{C},\mathcal{C}') := \sup\{ d(f,\mathcal{C}') \mid f \in \mathcal{C}\} \enspace,
\end{equation}
where, for each $f \in \mathcal{C}$, $d(f,\mathcal{C}') := \inf\{ d(f,g) \mid g \in \mathcal{C}'\}$ is the distance from $f$ to $\mathcal{C}'$.
(The symbol \lq r' stands for \lq\lq radius'', by an analogy as if $\mathcal{C}'$ is a single central point in the figure $\mathcal{C}$, in which case the $r$ is the radius of $\mathcal{C}$ in usual sense.)
\end{definition}
Among very various situations covered by Definition \ref{defn:function_density_problem} (where $\mathcal{C}$ in fact need \emph{not} even to be a set of functions), in the applications of FDPs discussed in this paper we will focus on the following typical cases:
\begin{definition}
[Function Density Problems -- typical cases]
\label{defn:function_density_problem_special}
Let $\mathcal{C}$ be the set of all functions $f \colon X \to Y$ from a given finite set $X$ to a given finite set $Y$.
Let $\mathcal{C}' \subset \mathcal{C}$.
For any $f,g \in \mathcal{C}$, we define the distance between $f$ and $g$ by
\begin{equation}
\label{eq:definition_distance}
d_{\mathrm{H}}(f,g) := |\{x \in X \mid f(x) \neq g(x)\}| \enspace.
\end{equation}
In this setting, a \emph{Function Density Problem} is a problem of estimating the quantity $r(\mathcal{C},\mathcal{C}')$ defined by \eqref{eq:function_density_problem__definition} with $d(\cdot,\cdot) = d_{\mathrm{H}}(\cdot,\cdot)$.
\end{definition}
In the case of Definition \ref{defn:function_density_problem_special}, the \lq\lq $\sup$'' and \lq\lq $\inf$'' in Definition \ref{defn:function_density_problem} can be simply replaced with \lq\lq $\max$'' and \lq\lq $\min$'', respectively.
Moreover, the distance defined by \eqref{eq:definition_distance} coincides with the (generalized) Hamming distance when members of $\mathcal{C}$ are identified with sequences of length $|X|$ over the alphabet $Y$ in a natural manner.
Note that the quantity $r(\mathcal{C},\mathcal{C}')$ can be regarded as a special case of so-called Hausdorff distance for two subsets of a metric space, which would support that it is reasonable to consider $r(\mathcal{C},\mathcal{C}')$.

An intuitive explanation of a motivation for the above definition is as follows.
Given a set $\mathcal{C}$ of functions, a subset $\mathcal{C}'$ consists of members of $\mathcal{C}$ which are in some sense \lq\lq easily analyzable'' or \lq\lq with simple descriptions''.
The distance $d(f,g)$ measures how two functions $f$ and $g$ are similar.
Then the quantity $d(f,\mathcal{C}')$ evaluates how accurately a function $f \in \mathcal{C}$ can be approximated by an \lq\lq easy'' function in $\mathcal{C}'$, and the quantity $r(\mathcal{C},\mathcal{C}')$ evaluates how densely the \lq\lq easy'' functions distribute among the entire set $\mathcal{C}$.
In other words, when $r(\mathcal{C},\mathcal{C}')$ is revealed to be small, it shows potential availability of a close approximation of any member of $\mathcal{C}$ by an \lq\lq easy'' function in $\mathcal{C}'$.
For example, in the case of Definition \ref{defn:function_density_problem_special}, any function $f \in \mathcal{C}$ can in principle be converted into some function $g \in \mathcal{C}'$ by changing the values $f(x)$ for at most $r(\mathcal{C},\mathcal{C}')$ points $x \in X$.
(We emphasize that it does \emph{not} mean that a close approximation of $f$ by a function in $\mathcal{C}'$ can be \emph{efficiently computable}.
Such a difference between existence and efficient computability is also relevant to a preceding observation for \lq\lq human ignorance'' by Rogaway \cite{Rog06}.)

\section{Hash Functions and FDPs}
\label{sec:hash_functions}

In this section, we point out a relation of FDPs introduced in Section \ref{sec:problem} to security analysis of keyless hash functions.
Here we propose a new framework for theoretical security evaluation of keyless hash functions based on FDPs.
Although theoretical security evaluation of keyless hash functions is evidently an extremely difficult problem and our proposed framework is unfortunately not yet practical, we hope that our framework can be a clue to this problem.

We consider a keyless hash function $H \colon X \to Y$ with possibly large but finite domain $X$ and relatively small (finite) range $Y$.
Among the major security requirements for hash functions, we focus on the collision resistance of $H$; we discuss how it is difficult to find a collision pair $(x_1,x_2)$ for $H$ (recall that $(x_1,x_2)$ is called a collision pair for $H$ if we have $x_1,x_2 \in X$, $x_1 \neq x_2$ and $H(x_1) = H(x_2)$).
To show the relevance of FDPs to this problem, first we give a somewhat informal description of an abstract \lq\lq typical'' strategy for finding a collision pair:
\begin{enumerate}
\item \label{item:strategy_hash_function__approximation}
Construct a close approximation $H' \colon X \to Y$ of $H$ in such a way that collision pairs for $H'$ can be found with reasonable computational time.
\item \label{item:strategy_hash_function__easy_collision}
Find randomly a collision pair $(x'_1,x'_2)$ for $H'$.
\item Construct from $(x'_1,x'_2)$ a candidate $(x_1,x_2)$ of a collision pair for $H$ (in the simplest case, we just set $(x_1,x_2) = (x'_1,x'_2)$).
\item Check if $(x_1,x_2)$ is a collision pair of $H$; if it is indeed a collision pair of $H$, then output $(x_1,x_2)$ and stop the process.
\item If $(x_1,x_2)$ is not a collision pair of $H$, go back to Step \eqref{item:strategy_hash_function__easy_collision} and repeat the process.
\end{enumerate}
Intuitively, the number of iterations in the above strategy before finding a collision pair for $H$ would be expected to be small if the approximation $H'$ is sufficiently close to $H$ (see Lemma \ref{lem:success_probability_attack_hash} below for a quantitative expression of this expected tendency).
Hence security of a hash algorithm $H$ against such an attack strategy is related to the possibility of finding its close approximation.

More precisely, we set $(x_1,x_2) = (x'_1,x'_2)$ in the above strategy for simplicity.
We consider the case of Definition \ref{defn:function_density_problem_special}, and let $\mathcal{C}'$ be a subset of $\mathcal{C}$ with the property that any hash function $H'$ in $\mathcal{C}'$ admits an efficient attack (finding a collision pair) by a certain known attack strategy.
In the above attack strategy, the approximation $H'$ for $H$ specified in Step \eqref{item:strategy_hash_function__approximation} is supposed to be chosen from $\mathcal{C}'$.
Now we have the following lemma:
\begin{lemma}
\label{lem:success_probability_attack_hash}
Suppose that $H$ and $H'$ are functions $X \to Y$ with $|Y| = n \geq 2$, and $d_{\mathrm{H}}(H,H') = d$, $0 < d < |X|$.
Then the probability that a collision pair for $H'$, which is chosen uniformly at random from the set of all collision pairs for $H'$, is also a collision pair for $H$ is not lower than
\begin{equation}
\label{eq:lem_success_probability_attack_hash}
\frac{ 2 \alpha_0 |X| - n(\alpha_0+1)\alpha_0 - 2d\alpha_0 }{ 2 \alpha_0 |X| + 2d|X| - n(\alpha_0+1)\alpha_0 - 2d\alpha_0 - d^2 - d } \enspace,
\end{equation}
where $\alpha_0 = \lfloor (|X| - d - 1)/n \rfloor$.
Moreover, when $|X| \geq d + (n-1)^2$, the value in \eqref{eq:lem_success_probability_attack_hash} is getting larger as $d$ becomes smaller.
\end{lemma}
A proof of Lemma \ref{lem:success_probability_attack_hash} will be provided in the last of this section.
Now let us imagine the following situation.
Two candidate sets $\mathcal{C}_1,\mathcal{C}_2$ for a new standard hash function are given, and we can specify subsets $\mathcal{C}'_1 \subset \mathcal{C}_1$ and $\mathcal{C}'_2 \subset \mathcal{C}_2$ in such a way that each $\mathcal{C}'_i$ ($i = 1,2$) consists of some hash functions for which collision pairs can be found in reasonable computational time by using some known techniques.
We suppose that $r(\mathcal{C}_1,\mathcal{C}'_1)$ is significantly small and $r(\mathcal{C}_2,\mathcal{C}'_2)$ is significantly large.
Then \emph{any} hash function $H$ chosen from $\mathcal{C}_1$ can be \emph{potentially} attacked by just finding a close approximation $H' \in \mathcal{C}'_1$ of $H$ (using some expert's sixth sense, for example) and applying the above attack strategy combined with known collision finding techniques.
On the other hand, $\mathcal{C}_2$ contains at least one hash function $H$ for which the above attack strategy combined with any known collision finding technique will not succeed.
This would suggest that it can be potentially safer to choose a new hash function from $\mathcal{C}_2$ rather than $\mathcal{C}_1$, as we already know the potential attack on any hash function in $\mathcal{C}_1$ but not the same for $\mathcal{C}_2$.

The authors hope that studies of FDPs can contribute to security analysis of keyless hash functions in the above manner, though how to specify the subset $\mathcal{C}'$ in practical cases is of course a big problem to be concerned.
One may also feel that it seems infeasible to compute the quantity $r(\mathcal{C},\mathcal{C}')$ for practical classes of hash functions; even if so, some estimate of a bound or tendency of $r(\mathcal{C},\mathcal{C}')$ would still give us an insight into the security level of those hash functions.
\begin{remark}
\label{rem:other_security}
Here we notice that, although we have focused on the collision resistance in the above argument, a similar idea would also be applicable to other security notions for keyless hash functions, such as the (second) preimage resistance.
\end{remark}

To conclude this section, we give a proof of Lemma \ref{lem:success_probability_attack_hash}.
\begin{proof}
[Proof of Lemma \ref{lem:success_probability_attack_hash}]
We write $(m)_2 := m(m-1)$ for any integer $m$.
Put $Y := \{y_1,\dots,y_n\}$, and for each $1 \leq i \leq n$, put
\begin{equation}
a_i := |\{x \in X \mid H'(x) = y_i\}| \,,\, b_i := |\{x \in X \mid H(x) \neq H'(x) = y_i\}| \enspace.
\end{equation}
Moreover, put
\begin{equation}
\varphi_1(\vec{a};\vec{b}) := \sum_{i=1}^{n} (a_i)_2 \,,\,
\varphi_2(\vec{a};\vec{b}) := \sum_{i=1}^{n} (a_i-b_i)_2 \enspace,
\end{equation}
where $\vec{a} := (a_1,\dots,a_n)$ and $\vec{b} := (b_1,\dots,b_n)$.
Then the number of collision pairs for $H'$ is $\varphi_1(\vec{a};\vec{b})$, while the number of collision pairs for $H$ is at least $\varphi_2(\vec{a};\vec{b})$.
Therefore the probability specified in the statement of Lemma \ref{lem:success_probability_attack_hash} is at least
\begin{equation}
\varphi(\vec{a};\vec{b}) := \frac{ \varphi_2(\vec{a};\vec{b}) }{ \varphi_1(\vec{a};\vec{b}) } \enspace.
\end{equation}
From now, we give a lower bound for the values of $\varphi$ under the following conditions implied by the definitions: $0 \leq b_i \leq a_i$ for each $i$, $\sum_{i=1}^{n} a_i = |X|$, and $\sum_{i=1}^{n} b_i = d$.
For the purpose, we show the following two lemmas:
\begin{lemma}
\label{lem:proof_lemma__success_probability_attack_hash__Step_1}
In the above setting, if the minimum value of the function $\varphi$ is attained by $\vec{a}$ and $\vec{b}$, then we have $b_i > 0$ for a unique index $i$, and $a_i - b_i \geq a_j$ for every index $j \neq i$.
\end{lemma}
\begin{proof}
If we have $i \neq j$ and $b_i,b_j > 0$, and we suppose $a_i \leq a_j$ by symmetry, then we have
\begin{equation}
\bigl( (a_i-1)_2 + (a_j+1)_2 \bigr) - \bigl( (a_i)_2 + (a_j)_2 \bigr)
= 2(a_j - a_i + 1) > 0 \enspace,
\end{equation}
therefore the value of $\varphi_1$ increases when $a_i$, $a_j$, $b_i$ and $b_j$ are replaced with $a_i-1$, $a_j+1$, $b_i-1$ and $b_j+1$, respectively.
On the other hand, the value of $\varphi_2$ is not changed by this replacement.
Therefore the value of $\varphi$ is decreased by this replacement, contradicting the assumption on the choice of $\vec{a}$ and $\vec{b}$.
Hence an index $i$ with $b_i > 0$ is unique, therefore $b_i = d$.
Similarly, if $j \neq i$ and $a_i - b_i < a_j$, then we have
\begin{equation}
\bigl( (a_i-b_i+1)_2 + (a_j-1)_2 \bigr) - \bigl( (a_i-b_i)_2 + (a_j)_2 \bigr)
= 2(a_i-b_i-a_j+1) \leq 0 \enspace,
\end{equation}
with equality holding when and only when $a_i - b_i = a_j - 1$.
This implies that the value of $\varphi$ at the $\vec{a}$ and $\vec{b}$ is larger than or equal to the value of $\varphi$ with $b_i$ and $b_j$ ($= 0$) being replaced with $b_i-1$ and $1$, respectively, where the equality holds if and only if $a_i - b_i = a_j - 1$.
As the former value is assumed to be the minimum, the equality condition $a_i - b_i = a_j - 1$ should hold.
Moreover, if $b_i - 1 > 0$, then the latter value of $\varphi$ (which is now equal to the former) cannot be the minimum by the above argument, which also leads to a contradiction.
Hence we have $b_i = 1$ (therefore $d = 1$) and $a_i = a_j$.
Now we have
\begin{equation}
\bigl( (a_i+1)_2 + (a_j-1)_2 \bigr) - \bigl( (a_i)_2 + (a_j)_2 \bigr)
= 2(a_i-a_j+1) > 0 \enspace.
\end{equation}
This implies that the value of $\varphi$ will decrease when $a_i$ and $a_j$ are replaced with $a_i + 1$ and $a_j - 1$, respectively, contradicting the assumption that the former value is the minimum.
Hence we have $a_i - b_i \geq a_j$ for every $j \neq i$, concluding the proof of Lemma \ref{lem:proof_lemma__success_probability_attack_hash__Step_1}.
\end{proof}
\begin{lemma}
\label{lem:proof_lemma__success_probability_attack_hash__Step_2}
In the above setting, if the minimum of the function $\varphi$ is attained by $\vec{a}$ and $\vec{b}$, then we have $|a_i - a_j| \leq 1$ for any pair of indices $i \neq j$ satisfying $b_i = b_j = 0$.
\end{lemma}
\begin{proof}
Assume contrary that $a_i - a_j \geq 2$ for such a pair of indices $i \neq j$.
For $\ell \in \{1,2\}$, let $\alpha_{\ell}$ denote the value of $\varphi_{\ell}$ at the $\vec{a}$ and $\vec{b}$, and let $\beta_{\ell}$ denote the value of $\varphi_{\ell}$ with $a_i$ and $a_j$ being replaced with $a_i-1$ and $a_j+1$, respectively.
Then we have $\beta_1 - \alpha_1 = \beta_2 - \alpha_2 = 2(a_j-a_i+1) < 0$.
On the other hand, for the unique index $i'$ with $b_{i'} > 0$ (see Lemma \ref{lem:proof_lemma__success_probability_attack_hash__Step_1}), we have $a_{i'} \geq b_{i'} + a_i \geq b_{i'} + a_j + 2 \geq 2$ by the assumption and Lemma \ref{lem:proof_lemma__success_probability_attack_hash__Step_1}, therefore $\alpha_1 > \alpha_2$.
Now we present the following lemma, which is proven by an easy calculation:
\begin{lemma}
\label{lem:inequality}
If $p > q \geq 0$ and $r > 0$, then $q/p < (q+r)/(p+r)$.
\end{lemma}
By using this lemma, we have
\begin{equation}
\frac{ \alpha_2 }{ \alpha_1 } = \frac{ \beta_2 - 2(a_j-a_i+1) }{ \beta_1 - 2(a_j-a_i+1) } > \frac{ \beta_2 }{ \beta_1 } \enspace,
\end{equation}
contradicting the assumption that $\alpha_2/\alpha_1$ is the minimum of the value of $\varphi$.
Hence Lemma \ref{lem:proof_lemma__success_probability_attack_hash__Step_2} holds.
\end{proof}

By Lemma \ref{lem:proof_lemma__success_probability_attack_hash__Step_1} and Lemma \ref{lem:proof_lemma__success_probability_attack_hash__Step_2}, the points $\vec{a}$ and $\vec{b}$ that attain the minimum of $\varphi$ satisfy the following conditions: $b_i > 0$ for a unique $i$, and there is an integer $\alpha$ satisfying that $a_i - b_i \geq \alpha + 1$ and $a_j \in \{\alpha,\alpha + 1\}$ for every $j \neq i$.
Note that this $\alpha$ can be taken as $\alpha \geq 0$; indeed, this is obvious if some $a_j$ with $j \neq i$ is positive, while the remaining possibility that $a_j = 0$ for every $j \neq i$ allows us to choose $\alpha = 0$ as $a_i = |X| > d = b_i$ and $a_i - b_i \geq 1$.
Let $k$ be the number of indices $j \neq i$ with $a_j = \alpha + 1$, therefore $0 \leq k \leq n-1$.
Then we have $a_i = |X| - (n - 1) \alpha - k$, while $b_i = d$, therefore the condition $a_i - b_i \geq \alpha + 1$ implies that $k \leq |X| - n\alpha - d - 1$.
Now we write the values of $\varphi_1$ and $\varphi_2$ in this case as $\varphi_1(\alpha,k)$ and $\varphi_2(\alpha,k)$, respectively.
Then we have
\begin{equation}
\begin{split}
\varphi_1(\alpha,k) &= k(\alpha+1)_2 + (n-1-k)(\alpha)_2 + (a_i)_2 \enspace, \\
\varphi_2(\alpha,k) &= k(\alpha+1)_2 + (n-1-k)(\alpha)_2 + (a_i-d)_2 \enspace,
\end{split}
\end{equation}
therefore $\varphi_1(\alpha,k) - \varphi_2(\alpha,k) = 2da_i - d^2 - d$.
Now by Lemma \ref{lem:inequality}, we have
\begin{equation}
\begin{split}
1 - \frac{ \varphi_2(\alpha,k) }{ \varphi_1(\alpha,k) }
= \frac{ 2da_i - d^2 - d }{ \varphi_1(\alpha,k) }
\leq \frac{ 2da_i - d^2 - d + 2d((n-1)\alpha + k) }{ \varphi_1(\alpha,k) + 2d((n-1)\alpha + k) } \\
= \frac{ 2d|X| - d^2 - d }{ k(\alpha+1)_2 + (n-1-k)(\alpha)_2 + (a_i)_2 + 2d(n-1)\alpha + 2dk }
\end{split}
\end{equation}
(note that $2d((n-1)\alpha + k) \geq 0$ as $\alpha \geq 0$).
Let $\psi(\alpha,k)$ denote the denominator of the right-hand side.
Then, by virtue of the property $\frac{\partial}{\partial k} a_i = -1$, we have
\begin{equation}
\frac{\partial}{\partial k} \psi(\alpha,k)
= (\alpha+1)_2 - (\alpha)_2 -(2a_i - 1) + 2d
= 2\alpha - 2a_i + 1 + 2d < 0
\end{equation}
(note that $a_i - d \geq \alpha + 1$), therefore $\psi(\alpha,k)$ is decreasing as $k$ is increasing.
On the other hand, we have $\psi(\alpha,n-1) = \psi(\alpha+1,0)$.
Now note that $\alpha \leq (|X| - d - 1)/n$ as $0 \leq k \leq |X| - n\alpha - d - 1$.
This implies that $\psi(\alpha,k)$ takes the minimum value at $\alpha = \lfloor (|X| - d - 1)/n \rfloor = \alpha_0$ and $k = k_0 := |X| - n\alpha_0 - d - 1$ (note that $k_0 \leq n-1$).
Moreover, we have $a_i = \alpha_0 + d + 1$ if $\alpha = \alpha_0$ and $k = k_0$.
Hence a straightforward calculation shows that
\begin{equation}
\begin{split}
1 - \frac{ \varphi_2(\alpha,k) }{ \varphi_1(\alpha,k) }
&\leq \frac{ 2d|X| - d^2 - d }{ \psi(\alpha_0,k_0) } \\
&= \frac{ 2d|X| - d^2 - d }{ 2 \alpha_0 |X| + 2d|X| - n(\alpha_0+1)\alpha_0 - 2d\alpha_0 - d^2 - d } \enspace,
\end{split}
\end{equation}
therefore
\begin{equation}
\begin{split}
\frac{ \varphi_2(\alpha,k) }{ \varphi_1(\alpha,k) }
&\geq 1 - \frac{ 2d|X| - d^2 - d }{ 2 \alpha_0 |X| + 2d|X| - n(\alpha_0+1)\alpha_0 - 2d\alpha_0 - d^2 - d } \\
&= \frac{ 2 \alpha_0 |X| - n(\alpha_0+1)\alpha_0 - 2d\alpha_0 }{ 2 \alpha_0 |X| + 2d|X| - n(\alpha_0+1)\alpha_0 - 2d\alpha_0 - d^2 - d } \enspace,
\end{split}
\end{equation}
which proves the lower bound \eqref{eq:lem_success_probability_attack_hash} in the statement of Lemma \ref{lem:success_probability_attack_hash}.

Finally, suppose that $d \geq 2$, and let $\eta_1(d)$ and $\eta_2(d)$ denote the denominator and the numerator in \eqref{eq:lem_success_probability_attack_hash}, respectively.
For any value $x$ depending on $d$, let $\Delta[x]$ temporarily denote the value of $x$ at $d-1$ minus the value of $x$ at $d$.
Then we have $\Delta(-d^2-d) = 2d$, therefore
\begin{equation}
\begin{split}
\Delta[\eta_2(d)] &= \Delta[2\alpha_0|X| - n(\alpha_0+1)\alpha_0 - 2d\alpha_0] \enspace, \\
\Delta[\eta_1(d)] &= \Delta[2\alpha_0|X| - n(\alpha_0+1)\alpha_0 - 2d\alpha_0] - 2|X| + 2d < \Delta[\eta_2(d)] \enspace.
\end{split}
\end{equation}
Moreover, we have $\Delta[\alpha_0] \in \{0,1\}$, and if $\Delta[\alpha_0] = 0$, then $\Delta[\eta_2(d)] = 2d\alpha_0 > 0$.
On the other hand, if $\Delta[\alpha_0] = 1$, then we have
\begin{equation}
\begin{split}
\Delta[(\alpha_0+1)\alpha_0] &= (\alpha_0+2)(\alpha_0+1)-(\alpha_0+1)\alpha_0 = 2(\alpha_0+1) \enspace, \\
\Delta[2d\alpha_0] &= 2(d-1)(\alpha_0+1) - 2d\alpha_0 = 2d - 2\alpha_0 - 2 \enspace,
\end{split}
\end{equation}
therefore
\begin{equation}
\begin{split}
\Delta[\eta_2(d)]
&= 2|X| - 2n(\alpha_0+1) - 2d + 2\alpha_0 + 2 \\
&= 2|X| - 2(n-1)\alpha_0 - 2n - 2d + 2 \\
&\geq 2|X| - 2(n-1) \frac{ |X| - d - 1 }{n} - 2n - 2d + 2 \\
&= \frac{ 2 }{ n } \bigl( |X| - d + 2n-1 - n^2 \bigr) \geq 0
\end{split}
\end{equation}
(where we used the assumption $|X| \geq d + (n-1)^2$).
Now by Lemma \ref{lem:inequality}, we have
\begin{equation}
\frac{ \eta_2(d-1) }{ \eta_1(d-1) }
= \frac{ \eta_2(d) + \Delta[\eta_2(d)] }{ \eta_1(d) + \Delta[\eta_1(d)] }
\geq \frac{ \eta_2(d) }{ \eta_1(d) + \Delta[\eta_1(d)] - \Delta[\eta_2(d)] }
> \frac{ \eta_2(d) }{ \eta_1(d) } \enspace.
\end{equation}
Hence the proof of Lemma \ref{lem:success_probability_attack_hash} is concluded.
\end{proof}

\section{PRGs and FDPs}
\label{sec:nb-PRG}

As our second application of FDPs, in this section we present some results which prove that any (computationally indistinguishable) PRG with some parameter is also an nb-PRG with a (possibly different) specified parameter.
The concrete relations between parameters for an algorithm as a PRG and as an nb-PRG, respectively, will be determined by applying FDPs.

First we recall the security notion for PRGs.
We emphasize that, for the sake of simplicity, here we adopt definitions in forms of concrete security rather than asymptotic security.
Let $U_X$ denote the uniform probability distribution over a finite set $X$.
\begin{definition}
[{see e.g., \cite{FSS07}}]
\label{defn:PRG_security}
Let $G \colon S \to X$ be an algorithm with finite input set $S$ and finite output set $X$.
Given parameters $T \geq 0$ and $\varepsilon \geq 0$, $G$ is called a \emph{$(T,\varepsilon)$-secure pseudorandom generator} (\emph{PRG}) if, for any algorithm (called a \emph{distinguisher}) $D \colon X \to \{0,1\}$ with time complexity bounded by $T$, we have $\mathsf{Adv}_D(G) \leq \varepsilon$ where $\mathsf{Adv}_D(G)$ denotes the \emph{advantage} of $D$ defined by
\begin{equation}
\mathsf{Adv}_D(G) := | Pr[D(U_X) = 1] - Pr[D(G(U_S)) = 1] | \enspace.
\end{equation}
\end{definition}

Let $\Delta(P_1,P_2)$ denote the statistical distance of two probability distributions $P_1,P_2$ over the same finite set $Z$ defined by
\begin{eqnarray}
\label{eq:statistical_distance_1}
\Delta(P_1,P_2) &:=& \frac{1}{2} \sum_{z \in Z} | Pr[P_1 = z] - Pr[P_2 = z] | \\
\label{eq:statistical_distance_2}
&=& \max_{E \subset Z} | Pr[P_1 \in E] - Pr[P_2 \in E] | \enspace.
\end{eqnarray}
Then the advantage $\mathsf{Adv}_D(G)$ of a distinguisher $D$ defined above is equal to $\Delta(D(U_X),D(G(U_S)))$, as both $D(U_X)$ and $D(G(U_S))$ are probability distributions over $\{0,1\}$.
This interpretation of the advantage gives us a motivation to enhance the above security notion of PRGs, as in the following definition introduced by Dubrov and Ishai \cite{DI06} (with slightly different formulation):
\begin{definition}
[{\cite{DI06}}]
\label{defn:nb-PRG_security}
Let $G \colon S \to X$ be an algorithm with finite input set $S$ and finite output set $X$.
Given parameters $T \geq 0$, $\varepsilon \geq 0$ and an integer $n \geq 2$, $G$ is called \emph{$(T,n,\varepsilon)$-secure} if, for any algorithm (distinguisher) $D \colon X \to \{0,1,\dots,n-1\}$ with time complexity bounded by $T$, we have $\mathsf{Adv}_D(G) \leq \varepsilon$ where we put $\mathsf{Adv}_D(G) := \Delta(D(U_X),D(G(U_S)))$.
Such an algorithm $G$ is called a \emph{PRG that fools non-boolean distinguishers} (\emph{nb-PRG}, in short).
\end{definition}
Note that $(T,2,\varepsilon)$-security is equivalent to $(T,\varepsilon)$-security in Definition \ref{defn:PRG_security}.
Several applications of nb-PRGs are discussed in \cite{DI06}.
For example, it was shown that randomness used in some kinds of \emph{information-theoretically secure} protocols (such as multi-party computation of certain types) can be replaced with outputs of nb-PRGs, without any restriction on computational complexity of the adversary against the protocol.
However, despite the significance of nb-PRGs mentioned above, it seems much more difficult to construct secure nb-PRGs than the case of usual PRGs against $1$-bit output distinguishers.
Indeed, to the authors' best knowledge, the only constructions of nb-PRGs in the literature so far are the ones by Dubrov and Ishai themselves in the original paper \cite{DI06}, and their construction is based on certain computational assumption which is less standard than those used in constructions of usual PRGs.
Hence, it is worthy to investigate a method to construct nb-PRGs (under standard computational assumptions).

Our proposal here is to establish a general theorem of the following form: Any $(T',\varepsilon')$-secure PRG is also a $(T,n,\varepsilon)$-secure nb-PRG, where the parameters $T'$ and $\varepsilon'$ as a usual PRG are determined by $T$, $n$ and $\varepsilon$ in a certain manner.
Such an implication result is evidently meaningful, as it enables us to convert a large number of existing PRGs under standard assumptions into nb-PRGs.
In fact, an implication relation as above has been mentioned (without proof) in \cite{DI06}.
Our aim here is to improve the preceding relation by introducing the idea of FDPs.

The above-mentioned relation is derived from the first expression \eqref{eq:statistical_distance_1} of statistical distance, in the following manner (which refers to a description in \cite{NH13}).
We introduce some notations.
Put $Y := \{0,1,\dots,n-1\}$ for simplicity.
For any subset $Z \subset Y$, let $\chi_Z \colon Y \to \{0,1\}$ denote the characteristic function of $Z$ defined by $\chi_Z(x) = 1$ if $x \in Z$ and $\chi_Z(x) = 0$ if $x \in Y \setminus Z$.
We write $\chi_z = \chi_{\{z\}}$ for simplicity when $Z = \{z\}$.
In this setting, for any PRG $G \colon S \to X$ and any non-boolean distinguisher $D \colon X \to Y$, the statistical distance $\Delta(D(U_X),D(G(U_S)))$ is equal to
\begin{equation}
\begin{split}
&\frac{1}{2} \sum_{y \in Y} | Pr[D(U_X) = y] - Pr[D(G(U_S)) = y] | \\
&= \frac{1}{2} \sum_{y \in Y} | Pr[\chi_y \circ D(U_X) = 1] - Pr[\chi_y \circ D(G(U_S)) = 1] | \\
&= \frac{1}{2} \sum_{y \in Y} \mathsf{Adv}_{\chi_y \circ D}(G) \enspace,
\end{split}
\end{equation}
where $\chi_y \circ D$ denotes an algorithm performed by first executing the distinguisher $D$ and then evaluating the output of $D$ by the function $\chi_y$.
An important property is that $\chi_y \circ D$ is a $1$-bit output algorithm, therefore it can be regarded as a distinguisher for the PRG $G$.
This implies that, to show that a $(T',\varepsilon')$-secure PRG $G$ is also a $(T,n,\varepsilon)$-secure nb-PRG, it suffices to choose the parameters as $T' = T + \delta_1$ and $\varepsilon' = 2 \varepsilon / n$, where $\delta_1$ is the maximum of the overhead in computational complexity of composing some $\chi_y$ ($y \in Y$) to $D$ (usually, $\delta_1$ can be set to be almost zero in practical situations).
In other words, we have the following proposition (which has been mentioned in \cite{DI06}):
\begin{proposition}
\label{prop:nb-PRG_implication_1}
In this setting, any $(T + \delta_1,2 \varepsilon / n)$-secure PRG is also $(T,n,\varepsilon)$-secure, where the quantity $\delta_1$ is defined in the above manner.
\end{proposition}
A drawback of this result is that, in practical applications the parameter $n$ (which is relevant to the allowable input size for an adversary against a protocol under consideration) should frequently be large, which makes the overhead in a bound of advantage in Proposition \ref{prop:nb-PRG_implication_1} too heavy.
We try to resolve the drawback by improving or modifying the above result.

Our first idea is to use the second expression \eqref{eq:statistical_distance_2} of statistical distance instead of the first one \eqref{eq:statistical_distance_1} used in the preceding argument.
Namely, in the same setting as above, the statistical distance $\Delta(D(U_X),D(G(U_S)))$ is equal to
\begin{equation}
\begin{split}
&\max_{Z \subset Y} | Pr[D(U_X) \in Z] - Pr[D(G(U_S)) \in Z] | \\
&= \max_{Z \subset Y} | Pr[\chi_Z \circ D(U_X) = 1] - Pr[\chi_Z \circ D(G(U_S)) = 1] | \\
&= \max_{Z \subset Y} \mathsf{Adv}_{\chi_Z \circ D}(G) \enspace.
\end{split}
\end{equation}
In the same way as Proposition \ref{prop:nb-PRG_implication_1}, the above argument implies the following result:
\begin{proposition}
\label{prop:nb-PRG_implication_2}
In this setting, any $(T + \delta_2,\varepsilon)$-secure PRG is also $(T,n,\varepsilon)$-secure, where $\delta_2$ is the maximum of the overhead in computational complexity of composing some $\chi_Z$ with $Z \subset Y := \{0,1,\dots,n-1\}$ to $D$.
\end{proposition}
In contrast to Proposition \ref{prop:nb-PRG_implication_1}, there exists no overhead for a bound of advantage $\varepsilon$ in Proposition \ref{prop:nb-PRG_implication_2}.
However, instead, the overhead $\delta_2$ for a bound of time complexity of distinguishers is expected to be too heavy, as the set $Y$ (of somewhat large size) may contain an extremely complicated subset $Z$, for which the computation of $\chi_Z$ would be inefficient.

From now, we try to improve the above-mentioned trade-off between overheads for bounds of advantage and of computational complexity, by applying the idea of FDPs.
Put $Y := \{0,1,\dots,n-1\}$ as above, and let $\mathcal{C}$ be the set of characteristic functions $\chi_Z \colon Y \to \{0,1\}$ for subsets $Z \subset Y$, and let $d = d_{\mathrm{H}}$ (see \eqref{eq:definition_distance}).
Then for $\chi_{Y_1},\chi_{Y_2} \in \mathcal{C}$, $d_{\mathrm{H}}(\chi_{Y_1},\chi_{Y_2})$ is equal to the size of the symmetric difference $Y_1 \ominus Y_2 := (Y_1 \setminus Y_2) \cup (Y_2 \setminus Y_1)$ of two subsets $Y_1$ and $Y_2$.
Now we fix a subset $\mathcal{C}'$ of $\mathcal{C}$.
Let $\delta_3$ be the maximum of the overhead in computational complexity of composing some $\chi_Z \in \mathcal{C}'$ to $D$.
Moreover, we put $r := r(\mathcal{C},\mathcal{C}')$ for simplicity.
Then we have the following result (we notice that, when $\mathcal{C}' = \{\chi_{\emptyset}\}$, the theorem gives almost the same result as Proposition \ref{prop:nb-PRG_implication_1}):
\begin{theorem}
\label{thm:nb-PRG_implication}
In the above situation, let $\delta_1$ be as specified in Proposition \ref{prop:nb-PRG_implication_1}.
If $G \colon S \to X$ is $(T + \delta_1,\varepsilon_1)$-secure and $(T + \delta_3,\varepsilon_3)$-secure, then $G$ is also $(T,n,r\varepsilon_1 + \varepsilon_3)$-secure.
\end{theorem}
\begin{proof}
For each distinguisher $D \colon X \to Y := \{0,1,\dots,n-1\}$, we write $\mu(Z) := Pr[D(U_X) \in Z]$ and $\mu'(Z) := Pr[D(G(U_S)) \in Z]$ for a subset $Z \subset Y$.
Let $Y_0$ be a subset of $Y$ that attains the maximum of the second expression \eqref{eq:statistical_distance_2} of the statistical distance;
\begin{equation}
\Delta(D(U_X),D(G(U_S))) = | \mu(Y_0) - \mu'(Y_0) | \enspace.
\end{equation}
Note that $Y_0$ can be chosen in such a way that $\mu(Y_0) - \mu'(Y_0) \geq 0$ (if this inequality fails, use $Y \setminus Y_0$ instead of $Y_0$), therefore
\begin{equation}
\label{eq:thm__nb-PRG_implication__subset_for_maximum}
\Delta(D(U_X),D(G(U_S))) = \mu(Y_0) - \mu'(Y_0) \enspace.
\end{equation}
%
%\begin{equation}
%\label{eq:thm__nb-PRG_implication__difference_is_positive}
%\mu(Y_0) - \mu'(Y_0) \geq 0
%\end{equation}
%(if this inequality fails, use $Y \setminus Y_0$ instead of $Y_0$).
Moreover, by the definition of $r$, there is a subset $Y_1 \subset Y$ satisfying that $\chi_{Y_1} \in \mathcal{C}'$ and $d_{\mathrm{H}}(\chi_{Y_0},\chi_{Y_1}) = | Y_0 \ominus Y_1 | \leq r$.
Now we have
\begin{equation}
\nu(Y_0) - \nu(Y_1) = \nu(Y_0 \setminus Y_1) - \nu(Y_1 \setminus Y_0) \mbox{ for each } \nu \in \{\mu,\mu'\} \enspace,
\end{equation}
therefore we have
\begin{equation}
\begin{split}
&(\mu(Y_0) - \mu'(Y_0)) - (\mu(Y_1) - \mu'(Y_1)) \\
={}&(\mu(Y_0) - \mu(Y_1)) - (\mu'(Y_0) - \mu'(Y_1)) \\
={}& (\mu(Y_0 \setminus Y_1) - \mu'(Y_0 \setminus Y_1)) - (\mu(Y_1 \setminus Y_0) - \mu'(Y_1 \setminus Y_0)) \enspace.
\end{split}
\end{equation}
Moreover, the right-hand side is equal to
\begin{equation}
\begin{split}
& \sum_{y \in Y_0 \setminus Y_1} (\mu(\{y\}) - \mu'(\{y\})) - \sum_{y \in Y_1 \setminus Y_0} (\mu(\{y\}) - \mu'(\{y\})) \\
\leq{}& \sum_{y \in Y_0 \ominus Y_1} | \mu(\{y\}) - \mu'(\{y\}) | \\
={}& \sum_{y \in Y_0 \ominus Y_1} | Pr[\chi_y \circ D(U_X) = 1] - Pr[\chi_y \circ D(G(U_S)) = 1] | \\
={}& \sum_{y \in Y_0 \ominus Y_1} \mathsf{Adv}_{\chi_y \circ D}(G) \enspace.
\end{split}
\end{equation}
Now if $D$ has computational complexity bounded by $T$, then the assumption on $G$ and the definition of $\delta_1$ imply that
\begin{equation}
\sum_{y \in Y_0 \ominus Y_1} \mathsf{Adv}_{\chi_y \circ D}(G)
\leq \sum_{y \in Y_0 \ominus Y_1} \varepsilon_1
= | Y_0 \ominus Y_1 | \cdot \varepsilon_1
\leq r \varepsilon_1 \enspace.
\end{equation}
Summarizing, we have
\begin{equation}
(\mu(Y_0) - \mu'(Y_0)) - (\mu(Y_1) - \mu'(Y_1)) \leq r \varepsilon_1 \enspace.
\end{equation}
This and \eqref{eq:thm__nb-PRG_implication__subset_for_maximum} implies that
\begin{equation}
\begin{split}
&\Delta(D(U_X),D(G(U_S))) \\
&= (\mu(Y_0) - \mu'(Y_0)) - (\mu(Y_1) - \mu'(Y_1)) + (\mu(Y_1) - \mu'(Y_1)) \\
&\leq r \varepsilon_1 + (Pr[D(U_X) \in Y_1] - Pr[D(G(U_S)) \in Y_1]) \\
&\leq r \varepsilon_1 + |Pr[\chi_{Y_1} \circ D(U_X) = 1] - Pr[\chi_{Y_1} \circ D(G(U_S)) = 1]| \\
&= r \varepsilon_1 + \mathsf{Adv}_{\chi_{Y_1} \circ D}(G) \leq r \varepsilon_1 + \varepsilon_3 \enspace,
\end{split}
\end{equation}
concluding the proof of Theorem \ref{thm:nb-PRG_implication}.
\end{proof}
Regarding the relation between parameters in Theorem \ref{thm:nb-PRG_implication}, first note that it is natural by the definitions to expect that $\delta_1 \leq \delta_3$, which allows us to suppose that $\varepsilon_1 \leq \varepsilon_3$.
Now let us imagine the following situation: We can find an appropriate subset $\mathcal{C}' \subset \mathcal{C}$ in such a way that every characteristic function $\chi_Z \in \mathcal{C}'$ has low computational complexity and the quantity $r := r(\mathcal{C},\mathcal{C}')$ is small.
In this case, $\delta_3$ can be small as well as $r$, and it would make the implication relation given by Theorem \ref{thm:nb-PRG_implication} more efficient than those in Propositions \ref{prop:nb-PRG_implication_1} and \ref{prop:nb-PRG_implication_2}, therefore the above-mentioned trade-off is improved.
Hence a study of FDPs (in particular, those for functions with $1$-bit output sets) will contribute to establish a better relation between PRGs and nb-PRGs.
\begin{remark}
\label{rem:risk-hedging}
We mention that, for the two applications of FDPs discussed in the last two sections, a kind of \lq\lq risk-hedging'' relation exists as follows.
Namely, if we find that the quantity $r(\mathcal{C},\mathcal{C}')$ tends to be large in general, then it would support the argument in Section \ref{sec:hash_functions} to show that keyless hash functions under consideration would have better security.
On the other hand, if we find that the quantity $r(\mathcal{C},\mathcal{C}')$ tends to be small in general, then it would support the argument in Section \ref{sec:nb-PRG} to show that overheads in parameters for nb-PRGs compared to PRGs would be practically small.
\end{remark}

At the last of this section, we give an example of the possible choices of the distinguished subset $\mathcal{C}'$ of $\mathcal{C}$.
We consider the case that $n = 2^{\ell}$ for an integer $\ell$; now the set $Y$ is identified with $\{0,1\}^{\ell}$ via the binary expression of integers, and each element $\chi_Z$ of $\mathcal{C}$ is regarded as a boolean function $\{0,1\}^{\ell} \to \{0,1\}$ with $\ell$-bit input.
In this setting, our task is to find a set $\mathcal{C}'$ of \lq\lq easy'' boolean functions with $\ell$-bit inputs, for which each function in $\mathcal{C}$ can be closely approximated by some function in $\mathcal{C}'$.
Here we consider $\ell$-variable \emph{disjunctive normal form} (\emph{DNF}) \emph{formulae}; recall that a \emph{literal} is $y_i$ or $\overline{y_i}$ (logical NOT of $y_i$) where $y_1,\dots,y_{\ell} \in \{0,1\}$ are input bits for a boolean function, a \emph{term} is a logical AND of literals, and a DNF formula is a logical OR of terms.
The number of terms in a DNF formula is called the \emph{size} of the formula.
It is straightforward to show that any $\ell$-bit input boolean function is equivalent to some DNF formula of size $2^{\ell}$.
Below we will define $\mathcal{C}'$ to be the set of ($\ell$-bit input) boolean functions which are either a logical XOR of at most $\ell$ DNF formulae of size significantly smaller than $2^{\ell}$, or the logical NOT of such a function.
Now the functions in $\mathcal{C}'$ is expected to be easier to compute than general boolean functions, hence the overhead $\delta_3$ in the bound for the distinguisher's computational complexity in Theorem \ref{thm:nb-PRG_implication} would be significantly better (i.e., smaller) than its counterpart $\delta_2$ in Proposition \ref{prop:nb-PRG_implication_2}, as desired.

Before specifying the sizes of the DNF formulae in the choice of $\mathcal{C}'$, first we give the following argument.
A boolean function $f \colon \{0,1\}^{\ell} \to \{0,1\}$ is called \emph{monotone}, if $y_i \leq y'_i$ for every $i \in \{1,\dots,\ell\}$ implies $f(y_1,\dots,y_{\ell}) \leq f(y'_1,\dots,y'_{\ell})$.
We note that $\chi_Z \in \mathcal{C}$ is monotone if and only if $Z \subset \{0,1\}^{\ell}$ is a filter (i.e., upper-closed set) of the $\ell$-dimensional lattice $\{0,1\}^{\ell}$.
Then we have the following result:
\begin{lemma}
\label{lem:boolean_function_to_monotone}
Any $\chi_Z \in \mathcal{C}$ with $\vec{0} := (0,0,\dots,0) \not\in Z$ can be expressed as the logical XOR of at most $\ell$ monotone boolean functions.
\end{lemma}
\begin{proof}
Let $Z_1$ denote the filter of $\{0,1\}^{\ell}$ generated by the minimal (with respect to the order of $\{0,1\}^{\ell}$) elements of $Z$.
Then we have $Z \subset Z_1$ and $\chi_Z = \chi_{Z_1} \oplus \chi_{Z_1 \setminus Z}$ where $\oplus$ denotes the logical XOR.
The claim holds if $Z_1 \setminus Z = \emptyset$.
On the other hand, if $Z_1 \setminus Z \neq \emptyset$, then we iterate the process to decompose $\chi_{Z_1 \setminus Z}$.
As the minimum weight of minimal elements of $Z_1 \setminus Z$ is strictly larger than the minimum weight of minimal elements of $Z$, and $\vec{0} \not\in Z$ by the assumption, the process terminates with at most $\ell$ steps, giving the decomposition $\chi_Z = \chi_{Z_1} \oplus \cdots \oplus \chi_{Z_{\ell'}}$ with $Z_i$ being filters of $\{0,1\}^{\ell}$, $1 \leq i \leq \ell' \leq \ell$.
Hence the claim holds.
\end{proof}

By the lemma, for any $\chi_Z \in \mathcal{C}$, either $\chi_Z$ or $\overline{\chi_Z} = \chi_{\{0,1\}^{\ell} \setminus Z}$ can be expressed as the logical XOR of at most $\ell$ monotone boolean functions.
On the other hand, for each of these monotone functions, recently Blais, H{\aa}stad, Servedio and Tan \cite{BHST14} showed the following result:
\begin{proposition}
[See \cite{BHST14}]
\label{prop:approximation_by_DNF}
For any $0 < \eta < 1$, every $\ell$-bit input monotone boolean function $f$ can be approximated by a DNF formula $g$ of size $2^{\ell - \Omega(\sqrt{\ell})}$ satisfying that $d(f,g) \leq \eta \cdot 2^{\ell}$ (where the dependence on $\eta$ is omitted in the $\Omega$ notation for simplicity).
\end{proposition}

Owing to the result, we define the subset $\mathcal{C}'$ of $\mathcal{C}$ as above, with the sizes of the DNF formulae satisfying the bound in Proposition \ref{prop:approximation_by_DNF}.
Then by Lemma \ref{lem:boolean_function_to_monotone} and Proposition \ref{prop:approximation_by_DNF}, we have $r = r(\mathcal{C},\mathcal{C}') \leq \ell \eta \cdot 2^{\ell}$, which is $o(2^{\ell})$ if $\eta = o(1/\ell)$.
This makes the trade-off in Theorem \ref{thm:nb-PRG_implication} better than Proposition \ref{prop:nb-PRG_implication_1} as desired, if $\varepsilon_3$ is not too larger than $r \cdot \varepsilon_1$.

We note that there are several results in the literature on approximations of boolean functions by not only DNF formulae but also those in other special classes.
For example, upper approximations of boolean functions (i.e., approximations of $f$ by $g$ satisfying that $f(y) \leq g(y)$ for every $y \in \{0,1\}^{\ell}$) by affine boolean functions were studied in e.g., \cite{HSSW10}.
The authors hope that the present work provides another motivation for the well-studied area of good approximations of boolean functions.

\section{Mathematical Examples of FDPs}
\label{sec:example_FDP}

This section is devoted to describe some examples for mathematical studies of FDPs themselves, rather than their cryptographic applications such as ones discussed in Sections \ref{sec:hash_functions} and \ref{sec:nb-PRG}.
The authors hope that one would feel that FDPs themselves are of independent interest as mathematical problems and mathematical studies of FDPs will be promoted.

\subsection{Vector spaces and their subspaces: A general bound}
\label{subsec:example_FDP_vector_spaces}

The examples of FDPs discussed below can be interpreted in the following manner.
The set $\mathcal{C}$ forms a finite-dimensional vector space over a finite field $\mathbb{F}$, with a distinguished basis $v_1,\dots,v_d$ where $d := \dim(\mathcal{C})$, hence each element of $\mathcal{C}$ admits a vector expression.
A subset $\mathcal{C}'$ is a linear subspace of $\mathcal{C}$, and the distance $d(f,g)$ is defined to be the (generalized) Hamming distance with respect to the vector expressions of $f,g \in \mathcal{C}$.
In this subsection, we show a general upper and lower bounds of the quantity $r(\mathcal{C},\mathcal{C}')$ in this case.
Namely, we have the following:
\begin{proposition}
\label{prop:upper_bound_by_codimension}
In the above setting, let $\ell$ denote the minimal integer $\ell'$ satisfying that $\sum_{i=0}^{\ell'} \binom{d}{i} (|\mathbb{F}|-1)^i \geq |\mathbb{F}|^{\mathrm{codim}_{\mathcal{C}}(\mathcal{C}')}$, where $\mathrm{codim}_{\mathcal{C}}(\mathcal{C}')$ denotes the codimension $d - \dim(\mathcal{C}')$ of $\mathcal{C}'$ in $\mathcal{C}$.
Then we have $\ell \leq r(\mathcal{C},\mathcal{C}') \leq \mathrm{codim}_{\mathcal{C}}(\mathcal{C}')$.
\end{proposition}
\begin{proof}
Put $d' := \dim(\mathcal{C}')$, therefore $\mathrm{codim}_{\mathcal{C}}(\mathcal{C}') = d - d'$.
First we prove the lower bound.
For each $w \in \mathcal{C}$ and $k \geq 0$, put $B(w,k) := \{w' \in \mathcal{C} \mid d(w,w') \leq k\}$.
Then we have $|B(w,k)| = \sum_{i=0}^{k} \binom{d}{i} (|\mathbb{F}| - 1)^i$.
On the other hand, by the definition of $r(\mathcal{C},\mathcal{C}')$, we have $\mathcal{C} \subset \bigcup_{w \in \mathcal{C}'} B(w,r(\mathcal{C},\mathcal{C}'))$.
This implies that
\begin{equation}
|\mathcal{C}| \leq |\mathcal{C}'| \cdot \sum_{i=0}^{r(\mathcal{C},\mathcal{C}')} \binom{d}{i} (|\mathbb{F}| - 1)^i \enspace,
\end{equation}
or equivalently $|\mathbb{F}|^{d-d'} = |\mathcal{C}| / |\mathcal{C}'| \leq \sum_{i=0}^{r(\mathcal{C},\mathcal{C}')} \binom{d}{i} (|\mathbb{F}| - 1)^i$.
Hence we have $\ell \leq r(\mathcal{C},\mathcal{C}')$ by the choice of $\ell$.

Secondly, we prove the upper bound.
By applying Gaussian elimination to any basis of $\mathcal{C}'$, it follows that there exist a basis $u_1,\dots,u_{d'}$ of $\mathcal{C}'$ and distinct indices $i_1,\dots,i_{d'} \in \{1,2,\dots,d\}$ with the property that, for each $1 \leq j \leq d'$, the coefficient of a basis element $v_{i_j}$ of $\mathcal{C}$ in $u_j$ is $1$ and the coefficient of $v_{i_j}$ in any other $u_k$ ($k \neq j$) is $0$.
Now for an arbitrary element $w = \sum_{i=1}^{d} c_i v_i \in \mathcal{C}$ ($c_i \in \mathbb{F}$), the above property of $u_1,\dots,u_{d'}$ implies that the distance between $w$ and $w' := \sum_{j=1}^{d'} c_{i_j} u_j \in \mathcal{C}'$ is at most $d - d'$, therefore $d(w,\mathcal{C}') \leq d - d'$.
Hence we have $r(\mathcal{C},\mathcal{C}') \leq d - d'$, concluding the proof of Proposition \ref{prop:upper_bound_by_codimension}.
\end{proof}
The next result shows how the lower and upper bounds in Proposition \ref{prop:upper_bound_by_codimension} are close to each other:
\begin{proposition}
\label{prop:example_bound_distance_tightness}
In the setting of Proposition \ref{prop:upper_bound_by_codimension}, we have
\begin{equation}
\begin{split}
\ell \leq \mathrm{codim}_{\mathcal{C}}(\mathcal{C}')
&\leq \log_{|\mathbb{F}|} \left( c_{\ell} (|\mathbb{F}|-1)^{\ell} d^{\ell} / \ell! \right) \\
&= \ell ( \log_{|\mathbb{F}|} (|\mathbb{F}| - 1) + \log_{|\mathbb{F}|} d ) + \log_{|\mathbb{F}|} c_{\ell} - \log_{|\mathbb{F}|} \ell! \enspace,
\end{split}
\end{equation}
where $c_{\ell} = \ell + 1$ if $\mathbb{F}$ is the two-element field $\mathbb{F}_2$, and $c_{\ell} = (|\mathbb{F}|-1)/(|\mathbb{F}|-2)$ otherwise.
\end{proposition}
\begin{proof}
It suffices to prove the second inequality.
As $|\mathbb{F}|^{\mathrm{codim}_{\mathcal{C}}(\mathcal{C}')} \leq \sum_{i=0}^{\ell} \binom{d}{i}(|\mathbb{F}|-1)^i$ by the definition of $\ell$, it suffices to show that $\sum_{i=0}^{\ell} \binom{d}{i}(|\mathbb{F}|-1)^i \leq c_{\ell} (|\mathbb{F}|-1)^{\ell} d^{\ell} / \ell!$, or more generally, $\sum_{i=0}^{m} \binom{N}{i}(q-1)^i \leq c'_m (q-1)^m N^m / m!$ for all integers $N \geq m \geq 0$ and $q \geq 2$, where we put $c'_m := (q-1)/(q-2)$ if $q \geq 3$ and $c'_m := m+1$ if $q = 2$, and we set $0^0 = 1$ (note that $\ell \leq \mathrm{codim}_{\mathcal{C}}(\mathcal{C}') \leq d$).
We use induction on $m$.
The case $m = 0$ is trivial.
For the case $m \geq 1$, we have
\begin{equation}
\begin{split}
\sum_{i=0}^{m} \binom{N}{i}(q-1)^i
&= \sum_{i=0}^{m-1} \binom{N}{i}(q-1)^i + \binom{N}{m}(q-1)^m \\
&\leq \frac{ c'_{m-1} (q-1)^{m-1} N^{m-1} }{ (m-1)! } + \binom{N}{m} (q-1)^m \qquad \\
&\leq \frac{ c'_{m-1} (q-1)^{m-1} N^{m-1} }{ (m-1)! } + \frac{ (q-1)^m N^m }{ m! } \\
&= \frac{ (q-1)^m N^m }{ m! } \left( \frac{ c'_{m-1} m }{ (q-1) N } + 1 \right) \enspace.
\end{split}
\end{equation}
By the relation $m \leq N$ and the definition of $c'_m$, we have
\begin{equation}
\frac{ c'_{m-1} m }{ (q-1) N } + 1 \leq \frac{ c'_{m-1} }{ q-1 } + 1 = c'_m \enspace,
\end{equation}
therefore the desired inequality holds for this $m$ as well.
Hence the claim of Proposition \ref{prop:example_bound_distance_tightness} holds.
\end{proof}

\subsection{Boolean functions of low degrees}
\label{subsec:example_FDP_Boolean_functions}

As a first concrete example, here we deal with the set $\mathcal{C}$ of the functions $X \to Y$ with $n$-bit inputs and $1$-bit outputs, i.e., we set $X := \{0,1\}^n$ and $Y := \{0,1\}$ (which is relevant to the situation of Section \ref{sec:nb-PRG}).
First note that, when we identify $\{0,1\}$ naturally with $\mathbb{F}_2$, each function $f \colon X \to Y$ can be expressed as an $n$-variable square-free polynomial;
\begin{equation}
f(x_1,\dots,x_n) = \sum_{\vec{a} = (a_1,\dots,a_n) \in \{0,1\}^n} f(\vec{a}) \chi_{\vec{a}}(x_1,\dots,x_n)
\end{equation}
where we put
\begin{equation}
\chi_{\vec{a}}(x_1,\dots,x_n) := \prod_{i;a_i = 0} (1 - x_i) \prod_{i;a_i = 1} x_i \mbox{ for } \vec{a} = (a_1,\dots,a_n)
\end{equation}
(note that $\chi_{\vec{a}}(x_1,\dots,x_n) = 1$ if $x_i = a_i$ for every $i$ and $\chi_{\vec{a}}(x_1,\dots,x_n) = 0$ otherwise, therefore $\chi_{\vec{a}}$ is indeed the characteristic function of $\vec{a} \in \{0,1\}^n$).
For example, when $n = 2$ we have
\begin{equation}
\begin{split}
f(x_1,x_2)
={}& f(0,0)(1-x_1)(1-x_2) + f(0,1)(1-x_1)x_2 \\
&+ f(1,0)x_1(1-x_2) + f(1,1)x_1x_2 \enspace.
\end{split}
\end{equation}
Now for each $0 \leq k \leq n$, we set $\mathcal{C}' = \mathcal{C}'_k$ to be the subset of $\mathcal{C}$ consisting of functions that can be expressed as a square-free polynomial of degree $\leq k$.
For example, $\mathcal{C}'_0$ is the set of constant functions, and $\mathcal{C}'_1$ is the set of affine functions.
The distance $d(f,g) = d_{\mathrm{H}}(f,g)$ is defined as in \eqref{eq:definition_distance}.
Note that changing the value of $f \in \mathcal{C}$ at a point $\vec{a} \in \{0,1\}^n$ is equivalent to adding the function $\chi_{\vec{a}}$ to the $f$.
In this situation, we have the following upper and lower bounds for the quantity $r(\mathcal{C},\mathcal{C}'_k)$:
\begin{proposition}
\label{prop:example_bound_distance}
In the above setting, put $u_{n,k} := \sum_{i=k+1}^{n} \binom{n}{i}$, and let $\ell_{n,k}$ be the minimum integer $\ell$ satisfying that $2^{u_{n,k}} \leq \sum_{i=0}^{\ell} \binom{2^n}{i}$.
Then we have
\begin{equation}
\ell_{n,k} \leq r(\mathcal{C},\mathcal{C}'_k) \leq \min\{u_{n,k},2^{n-1}\} \enspace.
\end{equation}
\end{proposition}
\begin{proof}
For the upper bound, note that $r(\mathcal{C},\mathcal{C}'_k) \leq 2^{n-1}$, as any function $f \in \mathcal{C}$ can be converted into a constant function by changing the value $f(x)$ at every point $x \in \{0,1\}^n$ with the property that $f(x)$ is in the minority among the $2^n$ values of $f$ (the number of such points is at most $2^{n-1}$).
Then the upper bound follows from Proposition \ref{prop:upper_bound_by_codimension}, as $\mathcal{C}$ is an $\mathbb{F}_2$-vector space of dimension $2^n$ and $\mathcal{C}'_k$ is its subspace of codimension $u_{n,k}$.
The lower bound also follows from Proposition \ref{prop:upper_bound_by_codimension}.
\end{proof}
By Proposition \ref{prop:example_bound_distance_tightness}, the quantities $\ell_{n,k}$ and $u_{n,k}$ in Proposition \ref{prop:example_bound_distance} satisfy the relation $\ell_{n,k} \leq u_{n,k} \leq n \ell_{n,k} + \log_2 (\ell_{n,k} + 1) - \log_2 \ell_{n,k}!$.
Table \ref{tab:lower_bounds} gives the precise values of $\ell_{n,k}$ for some smaller cases.
\begin{table}[htb]
\centering
\caption{The values of $\ell_{n,k}$ for some small parameters}
\label{tab:lower_bounds}
\begin{tabular}{cc|cccccccc}
&&\multicolumn{8}{c}{$n-k$} \\
&& $1$ & $2$ & $3$ & $4$ & $5$ & $6$ & $7$ & $8$ \\ \hline
    & $2$ & $1$ & $2$ & & & & & & \\
    & $3$ & $1$ & $2$ & $4$ & & & & & \\
    & $4$ & $1$ & $2$ & $4$ & $8$ & & & & \\
$n$ & $5$ & $1$ & $2$ & $5$ & $10$ & $16$ & & & \\
    & $6$ & $1$ & $2$ & $5$ & $13$ & $22$ & $32$ & & \\
    & $7$ & $1$ & $2$ & $6$ & $16$ & $31$ & $49$ & $64$ & \\
    & $8$ & $1$ & $2$ & $6$ & $19$ & $43$ & $75$ & $105$ & $128$
\end{tabular}
\end{table}

Here we introduce a geometric point of view to the above problem.
We introduce some notations.
For a subset $I \subset [n] := \{1,2,\dots,n\}$, put $x_I := \prod_{i \in I} x_i$, and let $a_I$ be the element $(a_1,\dots,a_n)$ of $\{0,1\}^n$ determined by $a_i = 1$ when and only when $i \in I$.
We write $\delta_I := \chi_{a_I}$ for simplicity.
Let $\Delta^{n-1}_+$ be the disjoint union of an isolated point $P$ and the standard $(n-1)$-simplex $\Delta^{n-1}$ on the vertex set $[n]$; we regard $P$ as \lq\lq the $(-1)$-dimensional face'' of $\Delta^{n-1}_+$.
For each $\emptyset \neq I \subset [n]$, let $\langle I \rangle$ denote the $(|I|-1)$-dimensional sub-simplex of $\Delta^{n-1}$ spanned by $I$, and let $\langle I \rangle^o$ be its relative interior (note that $\langle \{i\} \rangle^o = \langle \{i\} \rangle = \{i\}$ for each $i \in [n]$).
On the other hand, we put $\langle \emptyset \rangle = \langle \emptyset \rangle^o := P$.
Now for each function $f(x) = \sum_{I \subset [n]} c_I x_I$ ($c_I \in \mathbb{F}_2$), we define its \emph{geometric realization} $G_f$ by
\begin{equation}
G_f := \bigcup_{I;c_I = 1} \langle I^c \rangle^o \quad \mbox{(disjoint union),}
\end{equation}
where $I^c$ denotes the complement $[n] \setminus I$ of $I$ in $[n]$.
For each $I \subset [n]$, by the definition and the fact that $\delta_I = \sum_{J \supset I} x_J$ (recall that now the values of functions are in $\mathbb{F}_2$), $G_{\delta_I}$ is the (disjoint) union of $P$ and $\langle J \rangle^o$ for all $\emptyset \neq J \subset I^c$, therefore we have $G_{\delta_I} = P \cup \langle I^c \rangle$.
Moreover, for any $0 \leq k \leq n$ and $I \subset [n]$, we have $|I| \geq k+1$ if and only if $\langle I^c \rangle$ is at most $(n-k-2)$-dimensional.
This implies that a function $f \in \mathcal{C}$ belongs to $\mathcal{C}'_k$ if and only if $G_f$ does not intersect with the $(n-k-1)$-dimensional skeleton $\Delta^n_{n-k-1}$ of $\Delta^{n-1}_+$, which consists of the faces of $\Delta^{n-1}_+$ of dimension up to $n-k-1$.

Based on the above observation, we consider the following puzzle.
We imagine a situation that a lamp is associated to each face of $\Delta^{n-1}_+$.
A \emph{state} of $\Delta^{n-1}_+$ is a collection of light/dark properties of all the lamps.
Given a function $f$, the corresponding \emph{initial state} $\mathcal{I}_f$ is defined in such a way that a lamp at a face is light if and only if the relative interior of the face is contained in $G_f$.
At any state, the player of the puzzle is allowed to indicate a face $F$ of $\Delta^{n-1}_+$ (we call it \lq\lq \emph{push the face $F$}''), then the light/dark properties of lamps at $P$ and every sub-face of $F$ are flipped; such a process is regarded as a \emph{move} of the puzzle.
An initial state $\mathcal{I}_f$ is said to be \emph{solved} when the lamps of all faces of $\Delta^n_{n-k-1}$ are switched off by a sequence of moves started from $\mathcal{I}_f$.
With this interpretation, the distance $d(f,\mathcal{C}'_k)$ from $f \in \mathcal{C}$ to $\mathcal{C}'_k$ is the minimum of the number of moves to solve $\mathcal{I}_f$, and the quantity $r(\mathcal{C},\mathcal{C}'_k)$ is the minimal necessary number of moves to solve \emph{any} initial state.

Moreover, we also introduce a simplified puzzle on $\Delta^{n-1}$ instead of $\Delta^{n-1}_+$ by ignoring the isolated point $P$ in the above puzzle.
Let $r'_{n,k}$ denote the minimal necessary number of moves to solve (for the simplified puzzle) any initial state.
Then we have $r(\mathcal{C},\mathcal{C}'_k) = r'_{n,k} + 1$, as for an initial state $\mathcal{I}$ of the simplified puzzle for which solving $\mathcal{I}$ requires precisely $r'_{n,k}$ moves, one of the two initial states of the original puzzle obtained by adding a lamp at $P$ which is light and dark, respectively, requires $r'_{n,k} + 1$ moves.
Hence it suffices to consider the simplified puzzle on $\Delta^{n-1}$ for determining the quantity $r(\mathcal{C},\mathcal{C}'_k)$.
\begin{example}
\label{exmp:geometric_FDP_game}
We set $n = 4$ and show that $r(\mathcal{C},\mathcal{C}'_1) = 6$, or equivalently $r'_{4,1} = 5$.
Note that the general bounds in Proposition \ref{prop:example_bound_distance} only guarantee that $4 \leq r(\mathcal{C},\mathcal{C}'_1) \leq 8$ (note that $u_{4,1} = 11 > 8 = 2^{4-1}$).
We identify naturally each state in the puzzle on $\Delta^{n-1} = \Delta^3$ with each family of non-empty subsets of $[n] = [4]$, and we write $\{i_1,i_2,\dots,i_{\ell}\}$ as $i_1i_2 \cdots i_{\ell}$ for simplicity.
Moreover, to express each state we omit the subsets of $[4]$ of size larger than $2$, as the lamps at faces of dimension at least $n - k - 1 = 2$ are not relevant to determine whether the puzzle has been solved or not.
In other words, in the present situation, we can regard each state as edge and vertex coloring of the complete graph $K_4$.

First, we show that the initial state $\mathcal{I} = \{13,24\}$ requires more than $4$ moves to solve.
Assume contrary that $\mathcal{I}$ can be solved by at most $4$ moves.
If the player pushes the face $1234$, then a state $\{1,2,3,4,12,23,34,41\}$ is obtained.
To solve the state by at most $3$ remaining moves, the player has to push at least one $2$-dimensional face; we may assume by symmetry that the face is $123$.
Then the resulting state is $\{4,13,34,41\}$; however, a case-by-case analysis shows that to solve the state by at most $2$ remaining moves is impossible.
Therefore the player does not push the face $1234$.
On the other hand, if the player pushes a $2$-dimensional face, then we may assume by symmetry that the face is $123$, resulting in a state $\{1,2,3,12,23,24\}$.
To solve the state by at most $3$ remaining moves, the player has to push at least one more $2$-dimensional face.
If it is $124$, then we obtain a state $\{3,4,23,41\}$, but a case-by-case analysis shows that to solve the state by at most $2$ remaining moves is impossible (the case of $234$ is similar by symmetry).
If it is $134$, then we obtain a state $\{2,4,12,13,14,23,24,34\}$, but a case-by-case analysis shows that to solve the state by at most $2$ remaining moves is impossible as well.
Therefore the player does not push a $2$-dimensional face.
This implies that the player should push $13$ and $24$, resulting in a state $\{1,2,3,4\}$, from which to solve the state by at most $2$ remaining moves is impossible.
Hence we have a contradiction, therefore the initial state $S = \{13,24\}$ indeed requires more than $4$ moves to solve.

Secondly, we show that any initial state $\mathcal{I}$ can be solved by at most $5$ moves.
The player can solve $\mathcal{I}$ by at most $4$ moves when no lamps in $\mathcal{I}$ at $1$-dimensional faces are light, therefore $\mathcal{I}$ can be solved by at most $5$ moves when at most $1$ lamp in $\mathcal{I}$ at $1$-dimensional face is light.
When $2$ lamps in $\mathcal{I}$ at $1$-dimensional faces are light, a case-by-case analysis shows that $\mathcal{I}$ can be solved by at most $4$ moves unless $\mathcal{I}$ is of the form $\{i_1i_2,i_3i_4\}$ with $\{i_1,i_2\} \cap \{i_3,i_4\} = \emptyset$, and for any $\mathcal{I}$ of the latter form, $\mathcal{I}$ can be solved by pushing the faces $1234$, $i_1i_2i_3$, $i_1i_2i_4$, $i_1$, and $i_2$.
When $3$ lamps in $\mathcal{I}$ at $1$-dimensional faces are light, the problem can be reduced to the case of $2$ light lamps at $1$-dimensional faces by pushing one of the $3$ light lamps at $1$-dimensional faces.
When $4$ lamps in $\mathcal{I}$ at $1$-dimensional faces are light, the problem can be reduced to the case of $2$ light lamps at $1$-dimensional faces by pushing the face $1234$ unless $\mathcal{I}$ is of the form $\{1,2,3,4,i_1i_3,i_1i_4,i_2i_3,i_2i_4\}$ with $\{i_1,i_2\} \cap \{i_3,i_4\} = \emptyset$, and for any $\mathcal{I}$ of the latter form, $\mathcal{I}$ can be solved by pushing the faces $i_1i_2i_3$, $i_1i_2i_4$, $i_1$, and $i_2$.
When $5$ lamps in $\mathcal{I}$ at $1$-dimensional faces are light, the problem can be reduced to the case of $2$ light lamps at $1$-dimensional faces by pushing an appropriate $2$-dimensional face.
Finally, when $6$ lamps in $\mathcal{I}$ at $1$-dimensional faces are light, the problem can be reduced to the case of no light lamps at $1$-dimensional faces by pushing the face $1234$.
Hence any initial state $\mathcal{I}$ can be solved by at most $5$ moves, therefore we have $r'_{4,1} = 5$ as desired.
\end{example}

From now, we investigate FDPs in the above setting by using Gr\"{o}bner bases.
Recall that $X = \{0,1\}^n$.
Let $R := K[z_v \mid v \in X]$ be a polynomial ring in $2^n$ variables over a field $K$ of characteristic $0$.
We define the following ideal of $R$:
\begin{equation}
I_0 := (z_v{}^2 - 1 \mid v \in X) \subset R \enspace.
\end{equation}
For each $f \in \mathcal{C}$, put
\begin{equation}
z^f := \prod_{v \in X} z_v{}^{f(v)} \enspace.
\end{equation}
Then the set $\{z^f \mid f \in \mathcal{C}\}$ of all square-free monomials in $R$ forms a linear basis of the quotient ring $A_0 := R/I_0$.
Note that $z^f z^g = z^{f+g} \pmod{I_0}$ and the degree $\deg(z^f)$ of $z^f$ in $R$ is equal to $d(f,\underline{0})$ for any $f,g \in \mathcal{C}$, where $\underline{0}$ denotes the function in $\mathcal{C}$ taking constant value $0$.

Let $\mathcal{C}'$ be a subset of $\mathcal{C}$, which need not be a linear subspace of $\mathcal{C}$ unless otherwise specified.
We define the following ideal of $R$:
\begin{equation}
\check I_{\mathcal{C}'} := (z^f - z^g \mid f,g \in \mathcal{C}') \subset R \enspace,
\end{equation}
and consider the ideal $I_{\mathcal{C}'} := I_0 + \check I_{\mathcal{C}'}$ of $R$.
We identify the quotient ring $A_{\mathcal{C}'} := R/I_{\mathcal{C}'}$ with the quotient ring of $A_0$ by the image of $\check I_{\mathcal{C}'}$.
Now fix a graded monomial order, i.e., a monomial order $\prec$ satisfying that $\prod_{v \in X} z_v{}^{\alpha_v} \prec \prod_{v \in X} z_v{}^{\beta_v}$ for any exponents $(\alpha_v)_{v \in X}$ and $(\beta_v)_{v \in X}$ with $\sum_{v \in X} \alpha_v < \sum_{v \in X} \beta_v$.
Let $G$ be a Gr\"{o}bner basis for the ideal $I_{\mathcal{C}'}$, and consider the reduction process with respect to the Gr\"{o}bner basis $G$.
As each generator of $I_{\mathcal{C}'}$ is of the form \lq\lq $\mbox{(monic monomial)} - \mbox{(monic monomial)}$'', $G$ can be chosen in such a way that every element of $G$ is of the same form, and the linear basis of $A_0$ consisting of the square-free monic monomials can be partitioned into equivalence classes when projected onto the quotient ring $A_{\mathcal{C}'}$.
This also implies that the normal form $\mathrm{nf}(z^f)$ of each $f \in \mathcal{C}$ with respect to $G$ is a square-free monic monomial, i.e., of the form $z^g$ with $g \in \mathcal{C}$, and we have
\begin{equation}
\begin{split}
\deg(\mathrm{nf}(z^f))
&= \min\{ \deg(z^{f'}) \mid z^{f'} = z^{f} \pmod{I_{\mathcal{C}'}}\} \\
&= \min\{ \deg(z^{f'}) \mid z^{f'} - z^{f} \in I_{\mathcal{C}'}\} \enspace.
\end{split}
\end{equation}

Now we consider the case that $\underline{0} \in \mathcal{C}'$.
Note that $z^f z^f = 1 = z^{\underline{0}} \pmod{I_{0}}$ for any $f \in \mathcal{C}$.
Now if $f,g \in \mathcal{C}$ and $z^f = z^g \pmod{I_{\mathcal{C}'}}$, then we have $z^{f + g} = z^f z^g = z^f z^f = z^{\underline{0}} \pmod{I_{\mathcal{C}'}}$, therefore $f + g \in \mathcal{C}'$.
Conversely, if $f + g \in \mathcal{C}'$, then we have $z^f z^g = z^{f + g} = z^{\underline{0}} = 1 \pmod{I_{\mathcal{C}'}}$, therefore $z^f = z^f z^g z^g = z^g \pmod{I_{\mathcal{C}'}}$.
Hence $\deg(\mathrm{nf}(z^f))$ is equal to the minimal degree of $z^g$ with $g \in \mathcal{C}$ satisfying that $f + g \in \mathcal{C}'$, therefore $d(f,\mathcal{C}') = \deg(\mathrm{nf}(z^f))$.
This argument reduces the FDP in this setting to the problem of computing (the degrees of) the normal forms of square-free monomials.
More precisely, let $h_i$ denote the number of monic monomials in $A_{\mathcal{C}'}$ whose normal forms have degree $i$, and put $s := \max\{i \mid h_i > 0\}$.
(If the ideal is homogeneous, then $(h_i)_i$ is called the Hilbert function and it does not depend on the choice of a monomial order.)
Now the above argument implies that $r(\mathcal{C},\mathcal{C}') = s$.
Moreover, if $\mathcal{C}'$ is a linear subspace of $\mathcal{C}$, then we have $h_i \cdot |\mathcal{C}'| = |\{ f \in \mathcal{C} \mid d(f,\mathcal{C}') = i \}|$, therefore the data $(h_i)_i$ express the distributions of the distances $d(f,\mathcal{C}')$ over the functions $f \in \mathcal{C}$.

Based on the above argument, Proposition \ref{prop:upper_bound_by_codimension} can be restated for the present case as follows:
\begin{proposition}
\label{prop:bound_by_Grobner_basis}
In the above setting, suppose that $\mathcal{C}'$ is a linear subspace of $\mathcal{C}$.
Then we have
\begin{equation}
\min\{ \ell \mid \sum_{i=0}^{\ell} \binom{n}{i} \geq \frac{2^n}{|\mathcal{C}'|} \} \leq r(\mathcal{C},\mathcal{C}') \leq \frac{2^n}{|\mathcal{C}'|} \enspace.
\end{equation}
\end{proposition}
\begin{proof}
Note that the number of monic monomials in $A_{\mathcal{C}'}$ is $2^n / |\mathcal{C}'|$.
Then the lower bound follows from the fact that the normal form of each monic monomial is also a monic monomial and that there exist $\binom{n}{i}$ square-free monic monomials of degree $i$, hence $h_i \leq \binom{n}{i}$.
On the other hand, the upper bound is deduced from the fact that each divisor of a monic monomial of normal form is also of normal form, hence $h_i = 0$ if $h_j = 0$ and $j < i$.
This concludes the proof.
\end{proof}
For the case $\mathcal{C}' = \mathcal{C}'_k$ as discussed above, Table \ref{tab:example_Grobner} shows a calculation result of $r(\mathcal{C},\mathcal{C}'_k)$ and $(h_i)_i$ for small parameters $n$ and $k$, which is obtained by using computer algebra software \texttt{Singular/Sage}.
By the table, we have $r(\mathcal{C},\mathcal{C}'_k) = 6$ when $(n,k) = (4,1)$, as explained in Example \ref{exmp:geometric_FDP_game}.
Note that the values of $r(\mathcal{C},\mathcal{C}'_k)$ in Table \ref{tab:example_Grobner} are consistent with the lower bounds shown in Table \ref{tab:lower_bounds}.
\begin{table}[htb]
\centering
\caption{Computer calculation result for some small parameters}
\label{tab:example_Grobner}
\begin{tabular}{c|c|c|l}
$n$ & $k$ & $r(\mathcal{C},\mathcal{C}'_k)$ & $(h_i)_{i \geq 0}$ \\ \hline
$2$ & $1$ & $1$ & $(1,1)$ \\
$3$ & $1$ & $2$ & $(1,8,7)$ \\
$3$ & $2$ & $1$ & $(1,1)$ \\
$4$ & $1$ & $6$ & $(1,16,120,560,875,448,28)$ \\
$4$ & $2$ & $2$ & $(1,16,15)$ \\
$4$ & $3$ & $1$ & $(1,1)$
\end{tabular}
\end{table}

\subsection{Perfect codes and Reed--Solomon codes}
\label{subsec:example_FDP_RS-code}

In this subsection, we consider the case that $\mathcal{C}$ is an $n$-dimensional vector space over the $q$-element field $\mathbb{F}_q$, hence $\mathcal{C}$ is identified with $\mathbb{F}_q{}^n$, the distance $d(\cdot,\cdot)$ is defined to be the (generalized) Hamming distance (with respect to the vector expressions of elements), and $\mathcal{C}'$ is a linear subspace of $\mathcal{C}$ coming from the coding theory.
Let the subspace $\mathcal{C}'$ be an $(n,m,d)$-code, i.e., $\dim(\mathcal{C}') = m$ and the minimum distance of $\mathcal{C}'$ is $d$.
By the definition of minimum distance, we have the following well-known relation
\begin{equation}
\label{eq:example_minimum_distance}
q^m \sum_{i=0}^{\lfloor d/2 \rfloor} \binom{n}{i} (q-1)^i \leq q^n \enspace.
\end{equation}
This and the argument in Proposition \ref{prop:upper_bound_by_codimension} implies that $r(\mathcal{C},\mathcal{C}') \geq \lfloor d/2 \rfloor$.

We say that $\mathcal{C}'$ is a \emph{perfect code}, if the equality holds in \eqref{eq:example_minimum_distance}.
For a perfect code $\mathcal{C}'$, the above argument and Proposition \ref{prop:upper_bound_by_codimension} implies that $r(\mathcal{C},\mathcal{C}') = \lfloor d/2 \rfloor$.
For example, if $n = 2^k - 1$, $q = 2$ and $\mathcal{C}'$ is the Hamming code $H_k$ which is a $(2^k-1,2^k-k-1,3)$-code, then we have $r(\mathcal{C},\mathcal{C}') = 1$.
On the other hand, if $n = 23$, $q = 2$ and $\mathcal{C}'$ is the binary Golay code $G_{23}$ which is a perfect $(23,12,7)$-code, then we have $r(\mathcal{C},\mathcal{C}') = 3$.
(In the case of the extended Golay code $\mathcal{C}' = G_{24}$ which is a nearly perfect $(24,12,8)$-code, where we set $n = 24$ and $q = 2$, we also have $r(\mathcal{C},\mathcal{C}') = 4$ in a similar manner.)

As another concrete class of $\mathcal{C}'$ for which the quantity $r(\mathcal{C},\mathcal{C}')$ can be explicitly determined, from now we study the case of \emph{Reed--Solomon codes}, which is also an important class of linear codes.
We write $q = p^e$ with a prime number $p$ and an integer $e \geq 1$, and choose an integer $k$ with $1 \leq k < n$.
Take a primitive element $\alpha$ of $\mathbb{F}_q$, i.e., $\mathbb{F}_q{}^{\times} = \langle \alpha \rangle$.
Define a polynomial $G(x) \in \mathbb{F}_q[x]$ of degree $n-k$ by
\begin{equation}
G(x) := (x-1)(x-\alpha)(x-\alpha^2) \cdots (x-\alpha^{n-k-1}) \enspace.
\end{equation}
For any integer $j \geq 0$, let $P_j$ denote the set of polynomials in $\mathbb{F}_q[x]$ of degrees up to $j$, which is a $(j+1)$-dimensional $\mathbb{F}_q$-linear subspace of $\mathbb{F}_q[x]$.
We identify $P_{n-1}$ with $\mathcal{C}$ via the correspondence $\sum_{i=0}^{n-1} a_i x^i \mapsto \sum_{i=0}^{n-1} a_i v_i$, where $(v_0,\dots,v_{n-1})$ is a distinguished linear basis of $\mathcal{C}$.
Now we introduce the following two linear maps:
\begin{equation}
\varphi_{n,k} \colon P_{k-1} \to P_{n-1} \,,\, f(x) \mapsto G(x)f(x) \enspace,
\end{equation}
\begin{equation}
\psi_{n,k} \colon P_{n-1} \to \mathbb{F}_q{}^{n-k} \,,\, f(x) \mapsto (f(1),f(\alpha),f(\alpha^2),\dots,f(\alpha^{n-k-1})) \enspace.
\end{equation}
Let $\mathcal{C}'$ be the image of $\varphi_{n,k}$, which is a subspace of $\mathcal{C}$ (via the above identification $\mathcal{C} \simeq P_{n-1}$).
This $\mathcal{C}'$ is a Reed--Solomon code.
Note that $\mathcal{C}'$ coincides with the kernel of $\psi_{n,k}$.
Now we have the following result:
\begin{proposition}
\label{prop:example_RS-code}
In the above setting of Reed--Solomon code, we have $r(\mathcal{C},\mathcal{C}') = n - k$.
\end{proposition}
\begin{proof}
As $\dim(\mathcal{C}') = k$, the inequality $r(\mathcal{C},\mathcal{C}') \leq n - k$ follows from Proposition \ref{prop:upper_bound_by_codimension}.
From now, we show that $r(\mathcal{C},\mathcal{C}') \geq n - k$, or equivalently, there exists an element $u \in \mathcal{C}$ satisfying that $d(u,\mathcal{C}') \geq n-k$.

For each polynomial $f(x) \in P_{n-1}$, the condition $d(f(x),\mathcal{C}') \leq n-k-1$ is equivalent to the following: There exist indices $0 \leq \nu_1 < \nu_2 < \cdots < \nu_{n-k-1} \leq n-1$ and coefficients $c_j \in \mathbb{F}_q$ ($1 \leq j \leq n-k-1$) for which we have $f(x) - \sum_{j=1}^{n-k-1} c_j x^{\nu_j} \in \mathcal{C}' = \ker \psi_{n,k}$, or equivalently,
\begin{equation}
\label{eq:prop_example_RS-code__condition}
f(\alpha^i) = \sum_{j=1}^{n-k-1} c_j \beta_{\nu_j}{}^i \mbox{ for every } 0 \leq i \leq n-k-1 \enspace,
\end{equation}
where we put $\beta_{\nu_j} := \alpha^{\nu_j}$.
The condition \eqref{eq:prop_example_RS-code__condition} can be expressed as
\begin{equation}
\left( 
\begin{array}{c}
f(\alpha^0) \\ f(\alpha^1) \\ \vdots \\ f(\alpha^{n-k-1})
\end{array}
\right) = \left( 
\begin{array}{cccc}
\beta_{\nu_1}{}^0 & \beta_{\nu_2}{}^0 & \cdots & \beta_{\nu_{N-K-1}}{}^0 \\ 
\beta_{\nu_1}{}^1 & \beta_{\nu_2}{}^1 & \cdots & \beta_{\nu_{N-K-1}}{}^1 \\ 
\vdots & \vdots & \vdots & \vdots \\ 
\beta_{\nu_1}{}^{n-k-1} & \beta_{\nu_2}{}^{n-k-1} & \cdots & \beta_{\nu_{n-k-1}}{}^{n-k-1}
\end{array}
\right) \vec{c} \enspace,
\end{equation}
where $\vec{c}$ denotes the column vector ${}^t(c_1,c_2,\dots,c_{n-k-1})$.
For simplicity, let $B$ and $\vec{b}$ denote, respectively, the first $n-k-1$ rows and the last row of the above matrix; i.e., the above condition is written as
\begin{equation}
{}^t(f(\alpha^0), f(\alpha^1),\dots,f(\alpha^{n-k-1})) = \left( 
\begin{array}{c}
B \\ \vec{b}
\end{array}
\right) \vec{c} \enspace.
\end{equation}
Now, as $\alpha$ is a primitive element of $\mathbb{F}_q$, all $\beta_{\nu_j}$ are distinct with each other and hence $B$ is a Vandermonde matrix which is invertible.
Therefore, the condition \eqref{eq:prop_example_RS-code__condition} implies that $\vec{c} = B^{-1} \cdot {}^t(f(\alpha^0),f(\alpha^1),\dots,f(\alpha^{n-k-2}))$ and $f(\alpha^{n-k-1}) = \vec{b} \vec{c}$.
On the other hand, the latter condition is not satisfied when $f(\alpha^i) = 0$ for every $0 \leq i \leq n-k-2$ and $f(\alpha^{n-k-1}) \neq 0$, e.g., $f(x) = \prod_{i=0}^{n-k-2} (x - \alpha^i)$.
Hence this element $f(x) \in \mathcal{C}$ satisfies that $d(f(x),\mathcal{C}') \geq n-k$, as desired.
This concludes the proof of Proposition \ref{prop:example_RS-code}.
\end{proof}

\section{Concluding Remarks}
\label{sec:conclusion}

In this paper, we first specified a class of mathematical problems, which we call Function Density Problems.
Then we pointed out novel connections of Function Density Problems to theoretical security evaluations of keyless hash functions and to constructions of provably secure pseudorandom generators with some enhanced security property.
Our argument aimed at proposing new theoretical frameworks for these topics (especially for the former) based on Function Density Problems, rather than providing some concrete and practical results on the topics.
We also gave some examples of mathematical discussions on the problems, which would be of independent interest from mathematical viewpoints.

To conclude this paper, we discuss some possible directions of future works.
First, there exist some cryptographic protocols for which the constructions are motivated by some NP-complete/NP-hard problems, but actually the distributions of the problem instances in the protocols are somewhat biased, therefore it has not succeeded to prove the security of the protocols directly from the hardness of the underlying problems (e.g., McEliece cryptosystem and other code-based protocols relevant to decoding problem for random linear codes; knapsack cryptosystem relevant to Subset Sum Problem; etc.).
We hope that the idea of Function Density Problems can be applied to measure the closeness of the approximations of the underlying hard problems in those protocols.
Secondly, for the mathematical characteristics of Function Density Problems, it would be interesting to evaluate the computational difficulty of Function Density Problems (e.g., to prove, if possible, that Function Density Problems are NP-hard).
Moreover, as the examples of Function Density Problems in this paper are for the case that the subset $\mathcal{C}'$ of $\mathcal{C}$ forms a linear subspace, it would be also significant to study the other cases that $\mathcal{C}'$ is not a linear subspace of $\mathcal{C}$.

\section*{Acknowledgments}

A preliminary version of this paper was presented at The 6th International Workshop on Security (IWSEC 2011), November 8--10, 2011 \cite{NAKMN11}.
The authors would like to thank the anonymous referees of IWSEC 2011 and for other submissions of the work for their kind reviews and comments.
The authors would also like to thank Goichiro Hanaoka for his precious comment on potential applications of the subject of this paper discussed in Section \ref{sec:conclusion}.

\end{document}